\documentclass[letterpaper, 10 pt, conference]{ieeeconf}  % Comment this line out if you need a4paper

\IEEEoverridecommandlockouts                              % This command is only needed if
\usepackage{algorithm}
\usepackage[noend]{algorithmic}
\usepackage{graphicx}
\usepackage[usenames,dvipsnames]{xcolor}
\usepackage[font=small]{caption}
\usepackage{subfig, float, amsmath, amssymb, amsfonts, graphics, graphicx, enumitem, hhline}
\usepackage{listings,multicol}
\usepackage{lipsum}
\usepackage{hyperref}
\usepackage{dashrule}
\usepackage{dblfloatfix}
\usepackage{placeins}
\usepackage{balance}
\newenvironment{psmallmatrix}
  {\left[\begin{smallmatrix}}
  {\end{smallmatrix}\right]}

\usepackage{bbm}		%added for indicator function font - O.M.

\newcommand{\liren}[1]{{\color{black}#1}}
\newcommand{\lcss}[1]{{\color{black}#1}}
\newcommand{\revise}[1]{{\color{black}#1}}
\newcommand{\hl}[1]{{\color{black}#1}}

\usepackage{amsthm}

\newtheoremstyle{theoremdd}% name of the style to be used
  {\topsep}% measure of space to leave above the theorem. E.g.: 3pt
  {\topsep}% measure of space to leave below the theorem. E.g.: 3pt
  {}% name of font to use in the body of the theorem
  {0pt}% measure of space to indent
  {\bfseries}% name of head font
  {.}% punctuation between head and body
  { }% space after theorem head; " " = normal interword space
  {\thmname{#1}\thmnumber{ #2}\thmnote{ (#3)}}

\theoremstyle{theoremdd}

\newtheorem{defn}{Definition}
\newtheorem{prop}{Proposition}
\newtheorem{lem}{Lemma}

\title{\LARGE \bf
Scalable Zonotopic Under-approximation of Backward Reachable Sets for Uncertain Linear Systems
}

\author{ Liren Yang \hspace{1cm} Necmiye Ozay% <-this % stops a space
\thanks{The authors are with the Dept. of
       EECS,
       Univ. of Michigan, Ann Arbor, MI 48109
      {\tt\small \{yliren, necmiye\}@umich.edu}. 
}
}

\begin{document}
\maketitle
\thispagestyle{empty}
\pagestyle{empty}

\begin{abstract}
Zonotopes are widely used for over-approximating \textit{forward} reachable sets of uncertain linear systems for verification purposes. 
In this paper, we use zonotopes to achieve more scalable algorithms that under-approximate \textit{backward} reachable sets of uncertain linear systems for control design. The main difference is that the backward reachability analysis is a two-player game and involves Minkowski difference operations, but zonotopes are not closed under such operations. 
% which is computationally challenging for zonotopes. 
We under-approximate this Minkowski difference with a zonotope, which can be obtained by solving a linear optimization problem. We further develop an efficient zonotope order reduction technique to bound the complexity of the obtained zonotopic under-approximations. The proposed approach is evaluated against existing approaches using randomly generated instances and illustrated with several examples.
\end{abstract}

\section{Introduction}

For autonomous control systems, the control objectives need to be achieved robustly against system uncertainties. 
Central to many control synthesis techniques for uncertain systems is backward reachability analysis. 
Given an uncertain control system and a set $X_0$ of target states, the backward reachable set (BRS) consists of the states that can be steered into $X_0$ in finite time, regardless of the system uncertainties. 
Being able to compute such sets is important to design controllers with safety or reachability objectives \cite{bertsekas1972infinite,lygeros1999controllers}, and is one building block for achieving more complicated control tasks 
% \cite{belta2017formal,chen2018signal}. 
\cite{chen2018signal}. 
Whenever the exact computation is hard, an under-approximation can be still used to define a conservative control law. 
A variety of approaches exists in the literature, including polyhedral computation \cite{blanchini2008set}, interval analysis \cite{li2017invariance}, HJB method \cite{mitchell2005time} and polynomial optimization \cite{lasserre2015tractable}, just to name a few. 
% the BRSs are represented as polyhedrons \cite{}, union of interval \cite{}, level sets of functions obtained by solving the HJB equation \cite{} or polynomial optimization problems \cite{}. 
For linear dynamics with linear constraints, polyhedra can be used to represent the BRSs as they are closed under linear transformation, Minkowski addition and subtraction, and can be computed leveraging linear optimization tools. However, it is limited to low dimensional systems (typically, state dimension $\leq 4$)  due to an expensive quantifier elimination step.

One related problem is the forward reachability analysis, where we deal with uncertain system with \textit{no control inputs} (e.g., closed-loop systems), and compute the set of states that can be visited in the future from some initial state in a given set $X_0$. %, under a certain disturbance profile. 
Such forward reachable sets can be computed offline for verification and online for state prediction \revise{\cite{althoff2021set}}. 
Often times, the forward reachable sets are overestimated for robustness. 
% The forward reachable sets are often overestimated for robustness. 
For linear systems, a special polyhedron called zonotope is widely used to represent forward reachable sets thanks to the favorable complexity of applying linear transformations (for forward state evolution) and Minkowski additions (to account for additive uncertainty) to zonotopes (see, e.g., \cite{althoff2015introduction,girard2005reachability}). 
Algorithms that compute zonotopic forward reachable sets are more scalable than those dealing with general polyhedra.

One natural question is: for \textit{uncertain} linear dynamics, is there a way to reverse the time so that the efficient zonotopic set computation for forward reachability analysis can be directly adopted to compute BRSs? Unfortunately, this is not the case. The main reason is that there lacks a meaningful notion of two-player game in forward reachability analysis. 
%{\color{blue} (\cite{kurzhanskiy2011reach}, need to be discussed)}. 
In the forward case, there is only one player (i.e., the environment) picking the initial state and the system uncertainty, whereas in the backward case, there are two players (i.e., the controller and the environment) picking the control input and the uncertainty in turn \revise{(see Section 4.2 of \revise{\cite{mitchell2007comparing}})}. Particularly, the existence of the environment player leads to a Minkowski subtraction step in the sequential BRS computation, but zonotopes are not closed under Minkowski subtraction \cite{althoff2015computing}. 
\revise{Therefore, combining the idea of time-reversing and efficient computational tools for forward reachability (e.g., based on  zonotopes \cite{liebenwein2020compositional}, \cite{han2016enlarging} or polynomial zonotopes \cite{kochdumper2020computing})} were explored only for \textit{deterministic} systems, but using zonotopes for \textit{uncertain} systems' backward reachability, to the best of our knowledge, is still missing.

%{\color{blue} One work I would like to mention (but don't know how) is Ian Mitchel's paper. 
%They are essentially searching for a zontopic initial set, together with a control strategy, under which the $T$-step ($T$  given and fixed) forward reachable sets (as zonotopes) are all safe. 
%A simplified version of what they do would be: search for a open-loop control sequence of length $T$, s.t. the resulted forward reachable zonotopes are safe. But that is conservative because the control depends neither on the initial state, nor the later states. So they search for a control strategy with certain form. For each generator $g_{t,i}$ of the $t^{\rm th}$ reachable zonotope $Z_t$, they search for a control input $u_{t,i}$. If the actual state is in $Z_t$, it can be written as a combination of $g_{t,i}$, and the same combination of $u_{t,i}$ will be used.  I do not know how less conservative this will be.   
%% Since the controller is not determined a prior, the forward zonotopes they compute looks like something under two-player game at the first glance (and hence might be related to what we are doing with a time-reversing argument). But 
%Since the form of the strategy is fixed, I think this is still a one-player game.  
%%They treat the admissible control set at each step as an uncertainty set, so that a control input can be searched freely within those ``uncertainty sets'', whereas the resulted trajectories are still contained by the zonotopic forward reachable sets. 
%}

In this paper, we use zonotopes to represent and compute BRSs for uncertain linear systems. 
The key ingredient is an efficient way to under/over-approximate the Minkowski difference of two zonotopes by solving convex optimization problems. 
While the under approximation allows us to efficiently compute a subset of the BRS \lcss{without polyhedral projection}, the over-approximation can be used to quantify how conservative this subset is.  
% leads to a scalable way to under-approximate the  sequentially, the over-approximation can be used to quantify the conservatism of the under approximation. 
Different from \cite{althoff2015computing}, which manipulates a hyperplane-representation, our approach only deals with the generator-representations of zonotopes, and hence is more efficient and suitable for sequential computation, but at the cost of accuracy. The accuracy issue, however, is mitigated by the fact that our subtrahend zonotope represents the impact of uncertainties and is usually small compared to the minuend zonotope. Moreover, \cite{althoff2015computing} does not guarantee if the approximation is an inner one or an outer one.  
\lcss{We also leverage the linear encoding of zonotope-containment problems \cite{sadraddini2019linear} and derive an alternative approach for under-approximating the Minkowski difference between zonotopes. 
Theoretical analysis and experiments show that our approach scales differently from this alternative.  
}
In order to upper bound the complexity of each step of the computation, we further present a way to reduce the order of the obtained zonotopic BRSs. 
Zonotope order reduction is extensively studied (e.g., see \cite{kopetzki2017methods}, \cite{yang2018comparison} and the references therein), but our approach is different: we search for a lower order zonotope \textit{enclosed} by the given zonotope, whereas existing techniques, focusing on forward reachability analysis, all look for outer approximations. 
Our approach is evaluated with randomly generated zonotopes with different dimensions and orders, and its efficacy is illustrated with several examples. 
% a case study on aircraft position control (with a 6D lateral dynamics and a 6D longitudinal dynamics). {\color{red}10D example? also in abstract}

\section{Preliminaries}
% \noindent \textbf{Notations. }
Let $G = [g_1, g_2, \dots g_N] \in \mathbb{R}^{n\times N}$ be a set of \text{generators}, and $c \in \mathbb{R}^{n}$ be a \text{center vector}. A \text{zonotope} $Z$ with generator-representation (or G-rep) $(G, c)$ is defined to be the set
$\big\{c + \sum_{i=1}^N \theta_i g_i \mid \theta_i \in [-1,1], \ i = 1,2,\dots N\big\}$. 
%  $\big\{c +G\theta \mid \theta \in [-1,1], \ i = 1,2,\dots N\big\}$. 
With a slight abuse of notation, we will write $Z = (G, c)$. 
Let $H \in \mathbb{R}^{L\times n}$ and $h \in \mathbb{R}^L$, a \textit{polyhedron} with hyperplane-representation (or H-rep) $(H,h)$  is the set $\{x \in \mathbb{R}^n \mid Hx \leq h\}$. If polyhedron $X$ is bounded, $X$ is called a \text{polytope}. 
A set $V= \{x_1, x_2 \dots, x_M\}\subseteq \mathbb{R}^{n}$ is called a 
vertex-representation (or V-rep) of a polytope $X$ if $X$ is the convex hull of $V$, i.e., $X  = \text{cvxh}(V) : = \big\{ \sum_{j=1}^M \lambda_j x_j \mid \sum_{j=1}^M \lambda_j = 1, \lambda_j \in [0,1], j = 1,2,\dots, M\big\}$, \revise{where $\text{cvxh}$ denotes the convex hull}.  Let $A \in \mathbb{R}^{L \times n}$ and $X \subseteq \mathbb{R}^n$ be a set, $AX$ denotes the set $\{Ax \mid x \in X\}$. 

Let $X, Y \subseteq \mathbb{R}^n$ be two sets, the \text{Minkowski sum} of $X$ and $Y$, denoted by $X \oplus Y$, is  the set $\{x + y\mid x \in X, y \in Y\}$. 
Whenever $X = \{x\}$ is a singleton set, we will write $x + Y$ for $X \oplus Y$. 
The \text{Minkowski difference} of $X$ and $Y$, denoted by $X\ominus Y$, is defined to be $\{z \in \mathbb{R}^n \mid z + Y\subseteq X\}$. 
For the Minkowski arithmetics, we assume that the operations are done in order from left to right, except as specified otherwise by parentheses.  The following lemmas will be useful. 

\revise{
\begin{lem}\label{lem:Min}
% Let $X, Y\subseteq \mathbb{R}^n$ be convex, compact and nonempty, $X \oplus Y \ominus Y = X$. 
Let $X, Y, Z \subseteq \mathbb{R}^n$. 
\begin{itemize}[nolistsep]
    \item[i)] [\cite{li2019robustly}, Proposition 3.1, \cite{yang2020efficient}, Lemma 4] $X \ominus Y \oplus Z \subseteq X \oplus Z \ominus Y$, particularly, $X \ominus Y \oplus Y \subseteq X \subseteq X \oplus Y \ominus Y$.  
    \item[ii)] [From \cite{montejano1996some}] If $X$, $Y$ and $Z$ are convex, compact and nonempty, then $X \oplus Z = Y \oplus Z$ implies that $X = Y$. 
    \item[iii)] If $X$, $Y$ are convex, compact and  nonempty, then $X \oplus Y \ominus Y = X$. 
\end{itemize}
\end{lem}

\begin{proof}
To prove iii), note that, by i), $X \oplus Y \ominus Y \oplus Y = X \oplus Y$. Then applying item ii) yields $X \oplus Y \ominus Y = X$. 
\end{proof}
}

\begin{lem} \ [\liren{From \cite{girard2005reachability}}] Let $Z = (G, c) \subseteq \mathbb{R}^n$ be a zonotope. 
\begin{itemize}[nolistsep]
    \item[i)] $Z = \bigoplus_{i=1}^N Z_i$  where $Z_i = \big(\{g_i\}, c_i\big)$ s.t. $\sum_{i=1}^N c_i = c$.  
    \item[ii)]  Let $A \in \mathbb{R}^{L\times n}$, $AZ = \big(AG, Ac)$. 
    \item[iii)]  Let $Z' = (G', c')$, $Z \oplus Z' = ([G, G'], c + c')$. 
\end{itemize}
\end{lem}

\section{Backward Reachable Sets}
Consider a discrete-time system in the following form: 
\begin{align}
    x_{t+1} = A x_t + B u_t + E w_t + K, 
    \label{eq:sys}
\end{align}
where $x \in \mathbb{R}^{n_x}$ is the state, $u \in U \in \mathbb{R}^{n_u}$ is the control input and $w \in W  \in \mathbb{R}^{n_w}$ is the disturbance input. 
% Assume that $U$ and $W$ are polytopes, and given a polytopic set $X_0$ of target states, 
Given a set $X_0$ of target states, we want to compute (or to under-approximate, if exact computation is hard) the $k$-step backward reachable set $X_k$ of set $X_0$, defined recursively as  
\begin{align}
    \hspace{-1.5mm}X_{k+1} & =  \{x\in \mathbb{R}^{n_x}  \mid \exists u \in U: \forall w \in W:  \nonumber \\
& \ \ \ Ax + Bu + Ew + K\in X_k\},  \ \ k = 0,1,2 \dots \label{eq:cpre}
\end{align}
Set $X_k$ contains the states from where it is possible to reach the target set $X_0$ in \textit{exactly} $k$ steps. 
A weaker definition of the $k$-step BRS would require $X_0$ to be reached in \textit{no more than} $k$ steps, whose formal definition is similar to Eq. \eqref{eq:cpre} except for an extra ``$\cup X_k$'' at the end of the formula. 
Here, we adopt the stronger definition in Eq. \eqref{eq:cpre} for simplicity because the union operation may lead to non-convex sets.    
There exists slightly different notions of reachable sets \cite{kurzhanskiy2011reach}, depending on the order of the quantifiers. We will focus on under-approximating the set defined by \eqref{eq:cpre} while our approach applies in general. 

Suppose that set $U$, $W$, and $X_0$ are polytopes, and that the H-rep of $U$, $X_0$ and the V-rep of $W$ is known, one can compute $X_k$ as a polytope in H-rep, i.e., 
\begin{align}
    \hspace{-2.5mm}X_{k+1} &  = \textbf{Proj}_x\big(\{x\in \mathbb{R}^{n_x}, u \in U \mid\forall w_j \in V_W:  \nonumber \\ & \ \ \ Ax + Bu + Ew_j \in X_k\}\big), \ k = 0,1,2\dots, 
\end{align}
where $\textbf{Proj}_x(S) = \{x\mid \exists u: (x,u)\in S\}$ is the projection operation \revise{(e.g., see \cite{smith2016interdependence}, Proposition 1)}. 
% To that end, the H-rep of set $U$ and the V-rep of set $W$ are needed. 
However, polytope projection is time-consuming, which limits the use of this approach to low dimensional systems (typically $n_x \leq 4$). 
% In fact, when performing such computation for our aircraft lateral dynamics ($n_x = 6, n_u = 2$) with hyper-rectangular target set $X_0$, MPT3 tool box cannot even finish computing $X_2$ before it reports a error while projecting an 8D polytope to $\mathbb{R}^6$.  

% In this paper, we consider the problem of under-approximating the BRSs of system \eqref{eq:sys}   
In this paper, we consider under-approximating the BRSs of system \eqref{eq:sys} instead 
under the following assumptions. 
\begin{itemize}[nolistsep]
    \item[A1.] The target set is a zonotope (denoted by $Z_0$ hereafter), whose G-rep is known.
    \item[A2.] The disturbance set $W$ is a polytope, whose H-rep $(H,h)$ and V-rep $V$ are both known.
    \item[A3.] Matrix $A \in \mathbb{R}^{n_x \times n_x}$ is invertible. This assumption is true whenever Eq. \eqref{eq:sys} is obtained by time-discretizing an underlying continuous-time linear dynamics. 
\end{itemize}
Finding under-approximation of BRSs is useful in control problems with reachability objectives and  falsification problems against safety requirements \cite{chou2018using}.

\section{Solution Approach}
We explore the use of zonotopes in under-approximating the BRS $X_k$. This is based on i) the modest computational complexity of operations on zonotopes such as Minkowski addition and affine transformation, and ii) the fact that   \eqref{eq:cpre} can be re-written as follows using Minkowski arithmetic \revise{\cite{kurzhanskiy2011reach}}: 
\begin{align}
    X_{k+1} = \{x \in \mathbb{R}^{n_x} \mid Ax \in X_k \ominus EW \oplus -BU - K\}. \label{eq:cpreM}
\end{align}
In Eq. \eqref{eq:cpreM}, if $W = \{0\}$ and the term ``$\ominus EW$" were not there, then one could show inductively that, under assumption A1-A3, $X_{k+1}$ is a zonotope whose G-rep can be easily computed from the G-reps of $X_k$ and $U$ after Minknowski addition and linear transformation. 
Whenever $W$ is not a singleton set, the key step is to efficiently under and over approximate the Minkowski difference in Eq. \eqref{eq:cpreM} with zonotopes in their G-reps. Whereas the former leads to an inner approximation of $X_{k+1}$, 
% which is what we want to compute, 
the latter one can be used to quantify the conservatism of this inner approximation.  
% Assuming that $W = \{w\}$ and that $A$ is invertible, then $X_{k+1} = A^{-1}(X_k -Ew\oplus - BU - K)$. This computation  only involves in Minkowski addition and affine transformation. Hence for higher dimensional systems,  $X_{k+1}$ can be  efficiently computed as a zonotope if $X_k$, $U$ are both zonotopes. Whenever  $W$ is not a singleton set, the challenge is to efficiently compute a zonotopic under-approximation of $X_k \ominus EW$. In this write-up, we present an efficient algorithm for such under-approximation. To quantify the conservatism of the under-approximation, we also   
% where ``$\oplus$" denotes the Minkowski sum ($X \oplus Y : = \{x + y, X \in X, y \in Y\}$) and ``$\ominus $" is the Minkowski minus. 

\subsection{Zonotopic Inner/Outer Approximation of $Z\ominus EW$}
Let $Z = (G, c) \subseteq \mathbb{R}^{n_x}$ be a zonotope, where $G = [g_1, g_2, \dots, g_N]$. We formulate two optimization problems, one computes a zonotopic outer approximation $\overline{\mathfrak{Z}}(Z, EW)$, and the other computes a zonotopic inner approximation $\underline{\mathfrak{Z}}(Z, EW)$, of set $EW$ using $Z$ as a ``template''. The obtained outer/inner approximation are also in G-reps. Particularly, their generators are scaled versions of $Z$'s generators, i.e., in the form of $\alpha_i g_i$ for some $\alpha_i \in [0,1]$ (see Definition \ref{defn:align}). We then show that the Minkowski difference $Z \ominus \overline{\mathfrak{Z}}(Z, EW)$ and $Z \ominus \underline{\mathfrak{Z}}(Z, EW)$ can be done  element-wise via generator subtraction. This leads to an efficient way to inner/outer approximate $Z \ominus EW$ with zonotopes in G-reps.  
This technique will become our key ingredient of BRS under-approximation.

\begin{defn}\label{defn:align}
Let $Z = (G, c)$ and $Z'= (G', c')$ be zonotopes. $Z'$ is \textit{aligned with} $Z$ if $G = [g_1, g_2, \dots, g_N]$ and $G'  =[\alpha_1 g_1, \alpha_2 g_2, \dots, \alpha_N g_N]$ for some $\alpha_i \in [0,1]$. 
\end{defn}

% \noindent \textbf{Outer approximation of $EW$. } 
\subsubsection{Outer approximation of $EW$}  
Consider the following linear programming problem: 
\begin{align}
    \hspace{-5mm} \begin{array}{rl}
        \min_{\theta, \, \alpha, \, c}  &  \sum_{i=1}^N b_i \alpha_i \\
        \text{s.t.} & \forall w_j \in V: c + \sum_{i=1}^N \theta_{ij} g_i = E w_j \\
       		   & \ \ \ |\theta_{ij}| \leq \alpha_i \leq 1, \ i = 0,1,\dots N \\
			      % \end{array}
    \end{array}, 
    \tag{min-out}
    \label{eq:minout}
\end{align}
where $b_i >0$ are constants and $\theta$ and $\alpha$ are vectors aggregated from $\theta_{ij}$ and $\alpha_i$ respectively.  
The V-rep $V$ of the disturbance set $W$, which is available by Assumption A2, is used to formulate the above problem. Let $N$ be the number of generators in the template zonotope $Z$, $n_x$ be the dimension of the ambient space, and $M$ be the number of vertices in $V$. In the optimization problem \eqref{eq:minout},  
\begin{align}
\begin{array}{rl}
    \# \text{variables}  & \hspace{-2mm} = \mathcal{O}(MN + n_x), \\
    \# \text{constraints} & \hspace{-2mm}= \mathcal{O}\big(M(N + n_x)\big). 
\end{array}
\label{eq:bigO1}
\end{align}

\begin{prop}\label{prop:Zminout}
Let $(\theta, \overline{\alpha}, \overline{c})$ be the minimizer of the optimization problem \eqref{eq:minout}. Define $\overline{\mathfrak{Z}}(Z, EW) = \big([\overline{\alpha}_1 g_1, \overline{\alpha}_2 g_2, \dots, \overline{\alpha}_N g_N ], \overline{c}\big)$. 
% and  $\overline{Z}_{EW} = (\overline{\textbf{g}}, \overline{b})$. 
We have $EW \subseteq \overline{\mathfrak{Z}}(Z, EW)$.
\end{prop}

\begin{proof}
By the conditions in \eqref{eq:minout}, for any $i$ and $w_j \in V$, there exist $\theta_{ij} \in [-\overline{\alpha}_i,\overline{\alpha}_i]$ 
%  for $i = 1,2\dots,N$ 
s.t. $Ew_j = \overline{c} + \sum_{i=1}^N \theta_{ij}g_i $. 
Equivalently, there exist $\theta_{ij}' \in [-1,1]$ s.t. $Ew_j = \overline{c} + \sum_{i=1}^N \theta_{ij}' \overline{\alpha}_i g_i $. 
Hence $EV \subseteq \overline{\mathfrak{Z}}(Z, EW) = \big([\overline{\alpha}_1 g_1, \overline{\alpha}_2 g_2, \dots, \overline{\alpha}_N g_N], \overline{c}\big)$. 
It then follows that  $EW = E \text{cvxh}(V) = \text{cvxh}(EV)\subseteq \overline{\mathfrak{Z}}(Z, EW)$ from the convexity of zonotope $\overline{\mathfrak{Z}}(Z, EW)$. 
\end{proof}

In general, there does not exist a unique minimal (in the set inclusion sense) zonotopic outer approximation of $EW$ that aligns with the template zonotope $Z$. We hence minimize a weighted sum of $\alpha_i$'s. The weights $b_i > 0$ can be used for heuristic design to incorporate prior knowledge of disturbance set $W$. For example, when $W$ is a hyper-rectangle and $E \in \mathbb{R}^{n_x \times n_w}$ is full rank, we use  
% choose $b_i$ to be smaller for $g_i$ that aligns closer with the axis. 
% \begin{align}
$b_i = \Vert Tg_i\Vert_1 - \Vert T g_i \Vert_\infty$, 
% \label{eq:bi}
% \end{align}
where $T = (E^\top E)^{-1}E$ when $n_x \geq n_w$ and $T =E^\top (E E^\top)^{-1}$ otherwise. 
The idea is to encourage using generators that closely align with vector $Ee_p$, where $e_p$ is the $p^{\rm th}$ natural basis of vector space $\mathbb{R}^{n_w}$.  A similar criteria 
% as Eq. \eqref{eq:bi} 
was used for zonotope order reduction in \cite{girard2005reachability}.

% \noindent \textbf{Inner approximation of $EW$. } 
\subsubsection{Inner approximation of $EW$}  

Consider the following optimization problem:  
% in Eq. \eqref{eq:maxin0}, which can be easily shown equivalent to \eqref{eq:maxin}.
\begin{align}
%&    \begin{array}{rl}
%        \max_{\underline{\alpha}, \, \underline{b}}  &  \log \det \big(\text{diag}(\underline{\alpha})\big) \\
%       \text{s.t.} & \forall \theta_i \in [-\underline{\alpha}_i, \underline{\alpha}_i]: H (\underline{b} + \sum_i \theta_{i} g_i)\leq h, \\
%       & 0 \leq \underline{\alpha} \leq 1 
%    \end{array}  \label{eq:maxin0}\\     
%\nonumber  \\
%& \ \ \ \ \ \ \ \ \ \ \ \ \ \ \ \ \ \Updownarrow G : = [g_1, g_2, \dots g_N]  \nonumber
%\\
%\nonumber \\
     \begin{array}{rl}
        \max_{\alpha, \, c}  &  \textstyle\sum_{i=1}^N d_i \log(\alpha_i) \\ 
       \text{s.t.} & H c + |HG| \alpha\leq h\\
       & 0 \leq \alpha \leq 1 
    \end{array}, 
    \tag{max-in}
    \label{eq:maxin}
\end{align}
where \revise{$d_i \geq 0$ are constants} and 
$|HG|$ is a matrix obtained by taking element-wise absolute value of matrix $HG$. The H-rep $(H, h)$ of the disturbance set $W$, which is available by Assumption A2, is used to formulate the above problem. 
% Let $L$ be the number of rows in matrix $H$.
Suppose that $H$ has $L$ rows. 
In  \eqref{eq:maxin}, we have 
\begin{align}
\begin{array}{rl}
    \# \text{ variables} &  \hspace{-2mm} = \mathcal{O}(N + n_x),  \\
    \# \text{ constraints}&  \hspace{-2mm}= \mathcal{O}(N+L). 
\end{array}
% \#\text{variables} = \mathcal{O}(N + n_x),  \#\text{constraints} = \mathcal{O}(N+L). 
\end{align}

\begin{prop}
Let $(\underline{\alpha}, \underline{c})$ be the maximizer of optimization problem \eqref{eq:maxin}. Define $\underline{\mathfrak{Z}}(Z, EW) = \big([\underline{\alpha}_1 g_1, \underline{\alpha}_2 g_2, \dots, \underline{\alpha}_N g_N ], \underline{c}\big)$.  
% and  $\overline{Z}_{EW} = (\overline{\textbf{g}}, \overline{b})$. 
We have $\underline{\mathfrak{Z}}(Z, EW)\subseteq EW$.  
\end{prop}

\begin{proof}
We first show that, for $\alpha \geq 0$ and any $c$,  $Hc + |HG|\alpha \leq h$ if and only if 
\begin{align}
\forall \theta \in \textstyle\prod_{i=1}^N[-\alpha_i, \alpha_i]: H (c + \sum_{i=1}^N \theta_{i} g_i)\leq h,
\label{eq:prop21}
 \end{align}
where $\theta_i$ is the $i^{\rm th}$ element of $\theta$. Let $H_{\ell}$ and $h_{\ell}$ be the $\ell^{\rm th}$ row of $H$ and $h$ respectively.  Eq. \eqref{eq:prop21} is equivalent to 
\begin{align}
& \forall \ell \in \{1,2,\dots, L\}: 
\begin{array}{rl}
\max_{\theta} & H_\ell (\underline{c} + G\theta) \leq h_\ell \\
\text{ s.t. }& % \forall i = 1,2,\dots N: \\
% & \theta_i \in [-\alpha_i, \alpha_i] 
\theta \in  \textstyle\prod_{i=1}^N[-\alpha_i,\alpha_i]
\end{array} \\
& \ \ \ \ \ \ \ \ \ \  \ \ \ \ \ \ \ \ \ \ \ \ \ \Updownarrow \nonumber \\
& \forall \ell \in \{1,2,\dots, L\}:  H_\ell c + |H_\ell G|\alpha \leq h_\ell. 
\label{eq:prop22}
\end{align}
%\begin{align}
%\begin{array}{l}
% \forall \ell \in \{1,2,\dots, L\}: 
%\begin{array}{rl}
%\max_{\theta} & H_\ell (\underline{c} + G\theta) \leq h_\ell \\
%\text{ s.t. }& \forall i = 1,2,\dots N: \\
%&  \theta_i \in [-\alpha_i, \alpha_i] \\ 
%\end{array} \\
%%  \ \ \ \ \ \ \ \ \ \  \ \ \ \ \ \ \ \ \ \ \ \ \ \Updownarrow  \\
% \Leftrightarrow \ \ \forall \ell \in \{1,2,\dots, L\}:  H_\ell c + |H_\ell G|\alpha \leq h_\ell. \\
%\end{array}
%\label{eq:prop22}
%\end{align}
%% where $\prod_{i=1}^N[-\alpha_i,\alpha_i] = \{\theta \mid \theta_i \in [-\alpha_i, \alpha_i]\}$. 
Eq. \eqref{eq:prop22} is equivalent to $Hc + |HG|\alpha \leq h$. Therefore the maximizer $(\underline{\alpha}, \underline{c})$ satisfies Eq. \eqref{eq:prop21}, which implies 
\begin{align}
\forall \theta' \in \textstyle\prod_{i=1}^N[-1, 1]: H (\underline{c} + \sum_{i=1}^N \theta_{i}'  \underline{\alpha}_i g_i)\leq h.  
\end{align}
That is, $ \big([\underline{\alpha}_1 g_1, \underline{\alpha}_2 g_2, \dots, \underline{\alpha}_N g_N ], \underline{c}\big) \subseteq EW$. 
\end{proof}

Again, the maximal (in the set inclusion sense) inner approximation does not exist in general. Here we maximize the volume of a hyper-rectangle in $\mathbb{R}^N$, defined by $d_i$ and $\alpha$. 
\revise{In particular, as a heuristic, we pick $d_i = \Vert g_i \Vert$ for $i = 1,2,\dots, N$ throughout the paper}.

\subsubsection{Efficient Minkowski Difference between Aligned Zonotopes}  
Next, we show that the Minkowski difference amounts to element-wise generator subtraction when the subtrahend zonotope is aligned with the minuend zonotope. 
\begin{prop}\label{prop:ZMinkowskiminus}
\normalfont 
Let $Z = (G, c)$ and $Z'= (G', c')$ be zonotopes and suppose that $Z'$ is aligned with $Z$. Then $Z \ominus Z' = \big([(1-\alpha_1) g_1, (1-\alpha_2) g_2, \dots, (1 - \alpha_N) g_N], c - c'\big)$. 
\end{prop}

\revise{
\begin{proof}
Let $\Delta:  = \big([(1-\alpha_1) g_1, (1-\alpha_2) g_2, \dots, (1 - \alpha_N) g_N], c - c'\big)$. 
By Lemma \ref{lem:Min} iii), $\Delta= \Delta \oplus Z' \ominus Z'$ as $\Delta$, $Z'$ are convex, compact and nonempty. Also note that $\Delta \oplus Z' = Z$, hence $\Delta = Z\ominus Z'$. 
\end{proof}
}

We summarize this part by the following proposition. 
\begin{prop}\label{prop:MinkMinusOverUnder}
Let $Z$ be a zonotope and let $\overline{\mathfrak{Z}}(Z, EW)$, $\underline{\mathfrak{Z}}(Z, EW)$ be defined by solving \eqref{eq:minout}, \eqref{eq:maxin} respectively, then $Z \ominus \overline{\mathfrak{Z}}(Z, EW) \subseteq Z \ominus EW \subseteq Z \ominus \underline{\mathfrak{Z}}(Z, EW)$. Particularly,  $Z \ominus \overline{\mathfrak{Z}}(Z, EW)$ and $Z \ominus \underline{\mathfrak{Z}}(Z, EW)$ can be computed efficiently with generator-wise subtraction. 
\end{prop}

\begin{proof}
It follows from Proposition \ref{prop:Zminout}-\ref{prop:ZMinkowskiminus} and the fact that both $\overline{\mathfrak{Z}}(Z, EW)$ and $\underline{\mathfrak{Z}}(Z, EW)$ are aligned with $Z$. 
%  by construction. 
\end{proof}

\subsection{Approximation of Backward Reachable Sets}
We can compute a zonotopic over/under-approximation of the BRS $X_k$ recursively as follows: 
\begin{align}
    \underline{Z}_0 & = \overline{Z}_0 = Z_0, \label{eq:cpreZapprx1}\\
    \underline{Z}_{k+1} & = A^{-1} \big(\underline{Z}_k \ominus \overline{\mathfrak{Z}}(\underline{Z}_k, EW)\oplus - BU - K\big), \label{eq:cpreZapprx2} \\
    \overline{Z}_{k+1} & = A^{-1} \big(\overline{Z}_k \ominus \underline{\mathfrak{Z}}(\overline{Z}_k, EW)\oplus - BU - K\big). \label{eq:cpreZapprx3}
\end{align}

\begin{prop}
\normalfont Let $X_k$ be defined by Eq. \eqref{eq:cpre}, and $\underline{Z}_k, \overline{Z}_k$ be defined by Eq. \eqref{eq:cpreZapprx1}-\eqref{eq:cpreZapprx3}, we have $\underline{Z}_k \subseteq X_k \subseteq \overline{Z}_k$. 
\end{prop}
 
\begin{proof}
We prove this by induction.  
When $k = 0$, $\underline{Z}_0 = \overline{Z}_0 = Z_0 = X_0$ by \eqref{eq:cpreZapprx1}. 
Suppose that $\underline{Z}_k \subseteq X_k \subseteq \overline{Z}_k$, we have 
\begin{align}
\hspace{-7mm}\underline{Z}_k \ominus \overline{\mathfrak{Z}}(\underline{Z}_k, EW) & \subseteq \underline{Z}_k \ominus EW & (\text{Proposition }\ref{prop:MinkMinusOverUnder})  \nonumber \\
& \subseteq X_k \ominus EW.  & (\underline{Z}_k \subseteq X_k) \label{eq:prop5}
\end{align}
Combining Eq. \eqref{eq:prop5}, \eqref{eq:cpreZapprx1} and Eq. \eqref{eq:cpreM} yields $\underline{Z}_{k+1} \subseteq X_{k+1}$. Similarly, one can show $X_{k+1} \subseteq \overline{Z}_{k+1}$. 
%\begin{itemize}[nolistsep]
%\item [$1^\circ$] $k = 0$, $\underline{Z}_0 = \overline{Z}_0 = Z_0 = X_0$ by Eq. \eqref{eq:cpreZapprx1}. 
%\item [$2^\circ$] Suppose that $\underline{Z}_k \subseteq X_k \subseteq \overline{Z}_k$. We have 
%\begin{align}
%\hspace{-7mm}\underline{Z}_k \ominus \overline{\mathfrak{Z}}(\underline{Z}_k, EW) & \subseteq \underline{Z}_k \ominus EW & (\text{Proposition }\ref{prop:MinkMinusOverUnder})  \nonumber \\
%& \subseteq X_k \ominus EW.  & (\underline{Z}_k \subseteq X_k) \label{eq:prop5}
%\end{align}
%Combining Eq. \eqref{eq:prop5}, \eqref{eq:cpreZapprx1} and Eq. \eqref{eq:cpreM} yields $\underline{Z}_{k+1} \subseteq X_{k+1}$. Similarly, one can show $X_{k+1} \subseteq \overline{Z}_{k+1}$. 
%\end{itemize} \vspace{-4.2mm}
\end{proof}

% \begin{rem}\label{rem:computation}
Eq. \eqref{eq:cpreZapprx2}, \eqref{eq:cpreZapprx3} only involve Minkowski addition, linear transformation of zonotopes and Minkowski difference where the subtrahend zonotope is aligned with the minuend zonotope. The above three operations can be done efficiently with G-rep manipulations. 
The time for computing $\underline{Z}_k$ grows modestly with $k$ because the number of $\underline{Z}_k$'s generators, denoted by $N_k$, increases linearly with $k$. In fact, $N_{k+1} = N_k + N_U$ where $N_U$ is the (constant) number of generators of the zonotopic control input set $U$.  
In what follows, we introduce an order reduction technique to upper bound the time complexity of computing $\underline{Z}_k$. 
% \end{rem}

\subsubsection{Zonotope Order Reduction}\label{sec:red}
The order of an $n$-dimensional zonotope with $N$ generators is defined to be $N/n$. 
Zonotope order reduction problem concerns approximating a given zonotope with another one with lower order.  
% This problem is extensively studied in the literature as the increasing order is the bottleneck in many applications with zonotopes.  
Most of the existing techniques focus on finding outer approximations because zonotopes are typically used to overestimate forward reachable sets. Whereas in this paper, we find inner approximations using the following fact.  
\begin{prop}\label{prop:reduction}
Let $Z = \big(G= [g_1, g_2, \dots, g_N], c\big)$ be a zonotope. Define $G_1$ to be the matrix after removing arbitrary two columns $g_i$, $g_j$ from $G$ and appending $g_i + g_j$, and define $G_2$ to be the matrix after removing columns $g_i$, $g_j$ from $G$ and appending $g_i -g_j$. Then $Z_1 = (G_1, c) \subseteq Z$ and  $Z_2 = (G_2, c) \subseteq Z$. 
\end{prop}

\revise{
\begin{proof}
Let $\l(g_k) := \{\theta g_k \mid \theta \in [-1,1]\}$, then $Z = c + \bigoplus_{k=1}^N l(g_k)$, $Z_1 = c+ \bigoplus_{k\neq i,j}^N l(g_k) \oplus l(g_i + g_j)$. Since $l(g_i + g_j) = \{\theta g_1 + \theta g_2 \mid \theta \in [-1,1]\} \subseteq \{\theta_1 g_1 + \theta_2 g_2 \mid \theta_1, \theta_2 \in [-1,1]\}= l(g_i) \oplus l(g_j)$, $Z_1 \subseteq Z$. Similarly, $Z_2 \subseteq Z$. 
\end{proof}
}

Note that, in Proposition \ref{prop:reduction}, the number of generators of $Z_1$ (or $Z_2$) is fewer than that of $Z$ by one. Our zonotope order reduction procedure will keep replacing two generators $g_i, g_j$ by their combination (either $g_i + g_j$ or $g_i - g_j$) until the order of the resulting zonotope is small enough. Particularly,  we use the following heuristic to select $g_i, g_j$: 
\begin{align}
(i,j) = \text{arg\,}\text{min}_{1\leq i < j\leq N} \Vert g_i\Vert_2 \Vert g_j - \hat{g}_i g_j^\top \hat{g}_i \Vert_2, 
\end{align} 
where $\hat{g}_i = \tfrac{g_i}{\Vert g_i \Vert_2}$. 
% $r_{ij} = \sigma_2 /\sigma_1$ and $\sigma_1 \geq \sigma_2$ are the two singular values of matrix $[g_i, g_j]$. 
Then we will add $g_i + g_j$ if $\Vert (g_i + g_j)^\top G'\Vert_2 \geq \Vert (g_i - g_j)^\top G'\Vert_2$, and add $g_i - g_j$ otherwise, where $G'$ is the transpose of the right inverse of the generator matrix after removing columns $g_i$, $g_j$. The idea is to combine two generators that are either closely aligned or small in 2-norm, and the combined generator should be larger and more perpendicular to the remaining generators.

\subsubsection{Deriving Reachability Control Law using $\underline{Z}_k$}\label{sec:law}
Once zonotopic inner approximations $\underline{Z}_k$ of the BRSs are computed, checking if a state $x$ belongs to $\underline{Z}_k$ amounts to solving a linear program. 
Moreover, for any state $x \in \underline{Z}_{k+1}$, we can find a control input $u \in U(x, \underline{Z}_k)$ that brings $x$ to $\underline{Z}_{k}$ in one step, where  $U(x, \underline{Z}_k)$ is defined to be
\begin{align}
    & \{u \in U \mid \forall w \in W: Ax + Bu + Ew + K \in \underline{Z}_{k}\}  \nonumber \\
 = & \textbf{Proj}_{u}\left\{(u, \theta)\,  \Bigg\vert\, \begin{array}{l} Ax + Bu + K = \\ 
c^{(k)} + \textstyle\sum_{i=1}^{N_{k}} \theta_i g_i^{(k)}, \\ u \in U, \theta_i \in [-1,1]\end{array}\right\}, \label{eq:getu}
\end{align}
where $\big([g_1^{(k)}, g_2^{(k)}, \dots, g_{N_k}^{(k)}], c^{(k)}\big)$ is the G-rep of $\underline{Z}_k \ominus \overline{\mathfrak{Z}}(\underline{Z}_k, EW)$, which can be saved during the computation (see Eq. \eqref{eq:cpreZapprx2}). We do not need to explicitly perform the projection step in Eq. \eqref{eq:getu} as it is sufficient to find one $u \in U(x, \underline{Z}_k)$ by solving a  linear program.  
For any initial state $x_0 \in \underline{Z}_k$, iteratively applying $u_t \in U(x_{t},\underline{Z}_{k-t-1})$ yields a feedback control strategy, which generates a sequence  $u_0, u_1, \dots u_{k-1}$ and drives the initial state $x_0$ into the target set $\underline{Z}_0 = Z_0$ in precisely $k$-steps, regardless of the disturbance inputs.

%\subsubsection{Extension}
%If the disturbance $w_t$ is available to the controller when determining $u_t$, the backward reachable set may be defined similar to Eq. \eqref{eq:cpre}, but with  the order of the quantifiers reversed,  
%\begin{align}
%    \hspace{-1.5mm}X_{k+1} & =  \{x\in \mathbb{R}^{n_x}  \mid \forall w \in W:  \exists u \in U: \nonumber \\
%         & \ \ \ \ \ Ax + Bu + Ew + K\in X_k\}, \nonumber \\
%& = A^{-1}\big(X_k \oplus -BU \ominus EW - K\big), k = 0,1,2,\dots\label{eq:cpre_react}
%\end{align}
%A similar approach can be used to over/under-approximate the above set. 

\section{Evaluation \& Discussion}

\subsection{Comparisons}
% \begin{rmk}
We compare our approach for under-approximating $Z \ominus EW$ with two other methods: one by Althoff \cite{althoff2015computing} and one based on the work by Sadraddini and Tedrake \cite{sadraddini2019linear}. 
Whenever the disturbance set $W$ is a zonotope in its G-rep, $Z \ominus EW$ can be estimated by \cite{althoff2015computing}, but the result is not guaranteed to be an under-approximation.  
This approach 
% is reported to 
outperforms the exact computation but is still expensive due to an H-rep manipulation.
% Another alternative under-approximates $Z\ominus EW$ based on the linear encoding of zonotope-containment problems \cite{sadraddini2019linear}. This leads to the following linear program:
\lcss{Alternatively, using the linear encoding of zonotope-containment problems \cite{sadraddini2019linear}, we derive the   linear program below that under-approximates $Z\ominus EW$:}
\begin{align}
\begin{array}{rl}
\max_{\Gamma, \gamma, \alpha, c} & \textstyle\sum_{i=1}^N  \alpha_i \\
\text{s.t.}& [G_Z\text{diag}(\alpha), EG_W] = G_Z \Gamma \\
& c_Z - (c + E c_W) = G_Z \gamma\\
& \Vert [\Gamma, \gamma] \Vert_\infty \leq 1, \ 0 \leq \alpha \leq 1
\label{eq:Sadra}
\end{array}, 
\end{align}
where $(G_W, c_W)$ and $(G_Z, c_Z)$ are the G-reps of $W$ and $Z$ respectively. 
% where $c_W$ and $c_Z$ are the center of $W$ and $Z$ respectively, and $G_W$, $G_Z$ are the matrices, whose columns are the generators of $W$ and $Z$ respectively.
Similar to our approach, the solution of \eqref{eq:Sadra} also gives a zonotopic under-approximation $(G_Z\text{diag}(\alpha), c  )$ of $Z \ominus EW$ that aligns with the template $Z$. 
% $W = (\{g_{W,1},\dots g_{W,N_W}\}, c_W)$, $Z = (\{g_{Z,1},\dots, g_{Z,N}\}, c_Z)$, $G_W = [g_{W,1}, \dots, g_{W,N_W}]$, $G_Z = [g_{Z,1}, \dots, g_{Z,N}]$. 
The linear program \eqref{eq:Sadra} scales differently from \eqref{eq:minout}, which dominates the time of computing BRSs. % computation time of our approach. 
%  Minkowski difference under-approximation.
 Let $N_W$ and $N$ be the number of generators of $W$ and $Z$ respectively. For \eqref{eq:Sadra},  
\begin{align}
\begin{array}{rl}
    \# \text{variables} &  \hspace{-2mm}= \mathcal{O}\big(N(N+N_W) + n_x\big), \\
    \# \text{constraints}  &  \hspace{-2mm}= \mathcal{O}(N + n_x). 
\end{array}
\label{eq:bigO2}
\end{align}
\hl{The size of \eqref{eq:Sadra} is independent of the number of $W$'s vertices and grows with $N_W$, the number of generators of $W$. Thus \eqref{eq:Sadra} is more advantageous than  \eqref{eq:minout} whenever $W$ is a high dimensional zonotope with a small order.  } \revise{On the other hand, the number of variables in \eqref{eq:minout} is linear in $N$, whereas that in \eqref{eq:Sadra} is quadratic in $N$. }

We randomly generate about 2000 test cases, each case consists of a zonotope $Z \subseteq \mathbb{R}^{n_x}$, a hyper-rectangular $W \subseteq\mathbb{R}^{n_x}$ and a square matrix $E\in \mathbb{R}^{{n_x}\times {n_x}}$. The Minkowski difference  $Z \ominus EW$ is estimated using the three different methods. 
Fig. \ref{fig:time} shows the computation time w.r.t. the dimension and the order of zonotope $Z$.  Each dot represents the time for a specific case, and the surface is plotted with averaged values. 
All the experiments are run on a 1.80 GHz laptop with 16 GB RAM. 
% It can be seen that 
The computation time of Althoff's approach grows fast w.r.t. the order and the dimension of $Z$ (in fact, we could not finish running any one of the higher-order cases after hours). 
\hl{Our approach scales better with the order of $Z$, but still grows relatively fast with the dimension $n_x$ because the number of $W$'s vertices grows exponentially with $n_x$ since we choose $W$ to be hyper-rectangles in this example.
Somewhat surprisingly, the computation time of Sadraddini's approach grows very slowly w.r.t. the order and the dimension of $Z$. 
This is consistent with the big-O analysis: in the largest test case, $n_x = 10$ and $N = 100$, 
but $W$ has about $10^3$ vertices $(M = 1000)$. Hence (min-out) has approximately ten times more variables than \eqref{eq:Sadra}. }
\begin{figure}[h]
  \centering
  \includegraphics[width=2.9in]{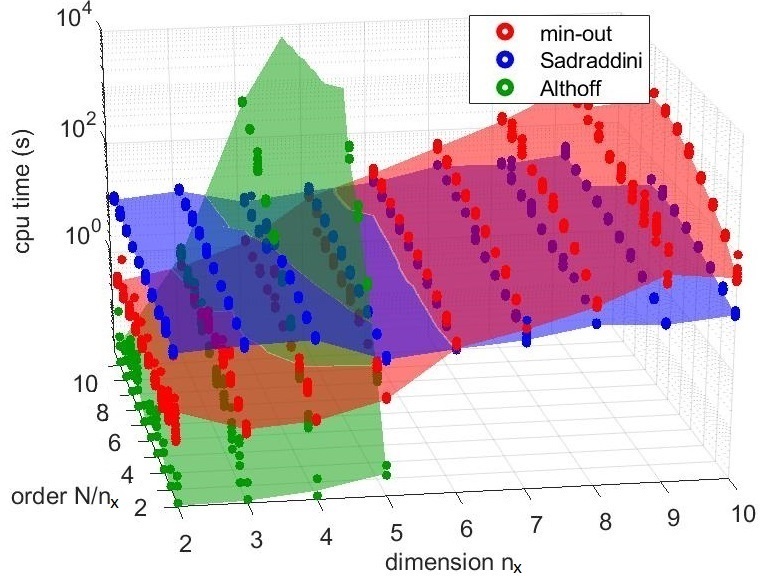}\\
  \caption{Upper: computation time for estimating $Z \ominus EW$. Lower: volume ratio distribution.}\label{fig:time}
\end{figure}
\lcss{Another metric is  the size of the obtained estimation. The volumes of the obtained zonotopic estimations are comparable.  Define $r_1 = \left(\tfrac{V_{\rm Althoff}}{V_{\rm min-out}}\right)^{1/n_x}$  and $r_1 = \left(\tfrac{V_{\rm Sadraddini}}{V_{\rm min-out}}\right)^{1/n_x}$, the statistics of $r_1$, $r_2$ are given in the table below. 
\begin{table}[h]
\vspace{-2mm}
\centering
\begin{tabular}{c|c|c|c|c|c}
\hline
 & mean & std. & min & max & confidence of $[0.95, 1.05]$\\ \hline
$r_1$ & $1.0017$ & $0.0577$ & $0.9900$ & $1.3856$ & $98.83\%$ \\ \hline
$r_2$ & $0.9678$ & $0.1891$  & $0.8372$ & $1.7498$ & $95.10\%$ \\ \hline
\end{tabular}
\end{table}
% \vspace{-7mm}
}

% Sadraddini's method tends to give larger estimation. 
%\begin{figure}[h]
%  \centering
%  \includegraphics[width=3.2in]{figs/plotvol.png}\\
%  \caption{Volume ratio distribution.}\label{fig:vol}
%\end{figure}

% {\color{blue} Another thing I thought of evaluating (but may be difficult) is, how often Althoff's method gives non-under-approximation. But N\&S condition in G-rep for zonotope containment is not known, I think.}

\subsection{Order Reduction}
We evaluate our order reduction technique with 29000 randomly generated zonotopes with different dimensions and orders. 
The approach introduced in Section \ref{sec:red} is used to reduce the order of each testing zonotope by one. 
As shown in Fig. \ref{fig:redtime} (upper), the computation time grows modestly with the zonotope's dimension and order. 
\begin{figure}[h]
  \centering
  \includegraphics[width=2.9in]{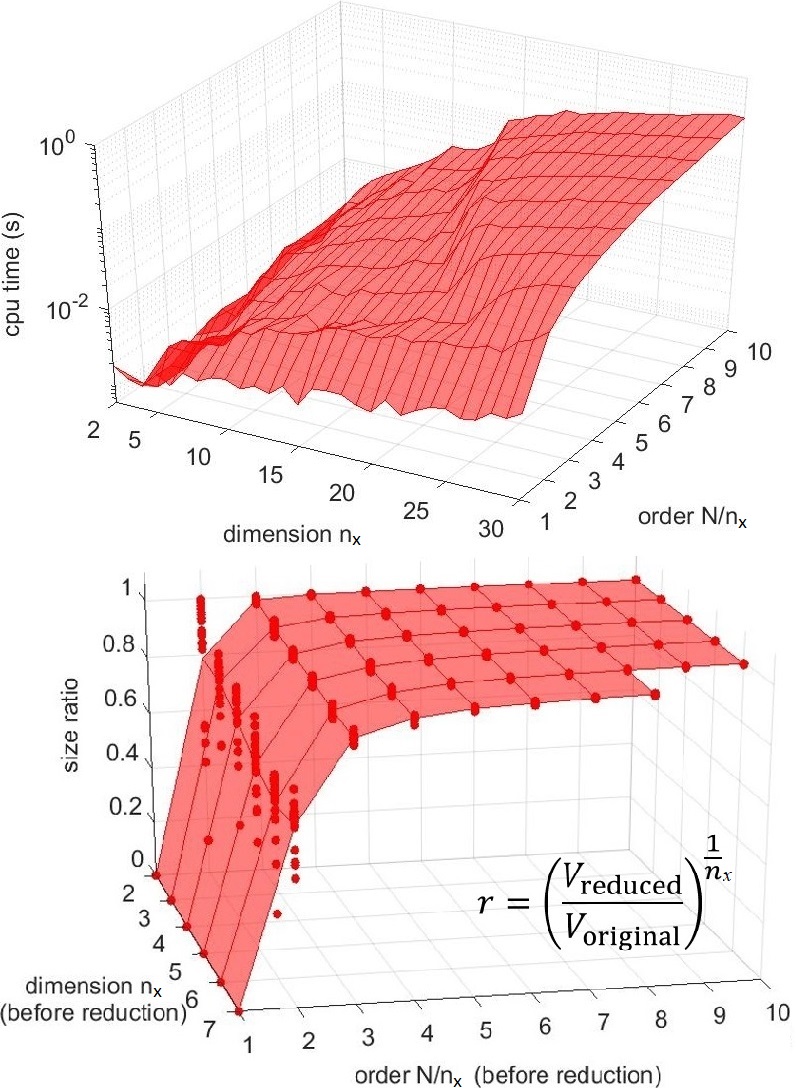}\\
  \caption{Upper: averaged computation time for reducing a zonotope's order by one. Lower: volume ratio between the reduced-order zonotope and the one before reduction.}\label{fig:redtime}
\end{figure}
The quality of the reduced-order zonotope is measured by the ratio between its volume and that of the original zonotope before reduction, defined in Fig. \ref{fig:redtime} (lower). We are able to run this evaluation for lower-dimensional cases because computing the exact volume of a zonotope is difficult for high-dimensional case due to the combinatorial complexity \cite{gover2010determinants}. % Note that in Fig. \ref{fig:red}, the $x$, $y$ coordinates are the dimension and the order of the original zonotope, hence reducing the order of a first order zonotope by one leads to empty set and zero ratio. 
In Fig. \ref{fig:redtime} (lower), the volume ratio increases with the the original zonotope's order because higher order means more freedom in selecting the generators to combine. 
In the presented cases, the ratio is close to one if the original zonotope's order is greater than three. 

\section{Case Studies}
\begin{figure*}[h]
  \centering
  \includegraphics[width=6.4in]{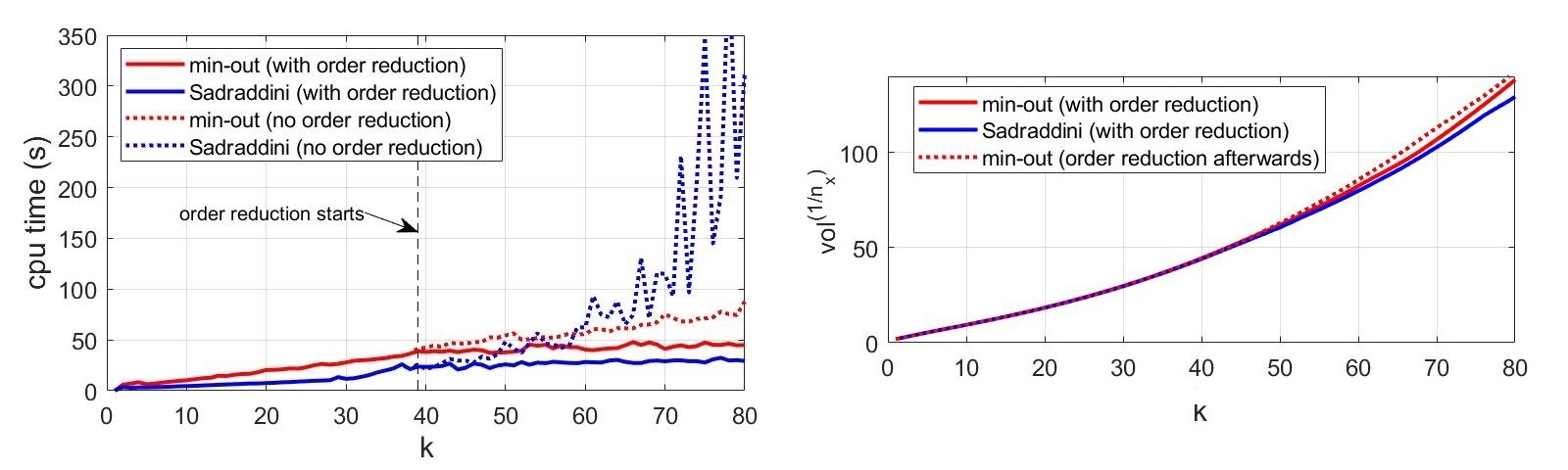}\\
  \caption{Backward reachable set computation for lateral dynamics. Left: computation time. Right: set volume.}\label{fig:lat}
%\end{figure}
%\begin{figure}[h]
%  \centering
  \includegraphics[width=6.4in]{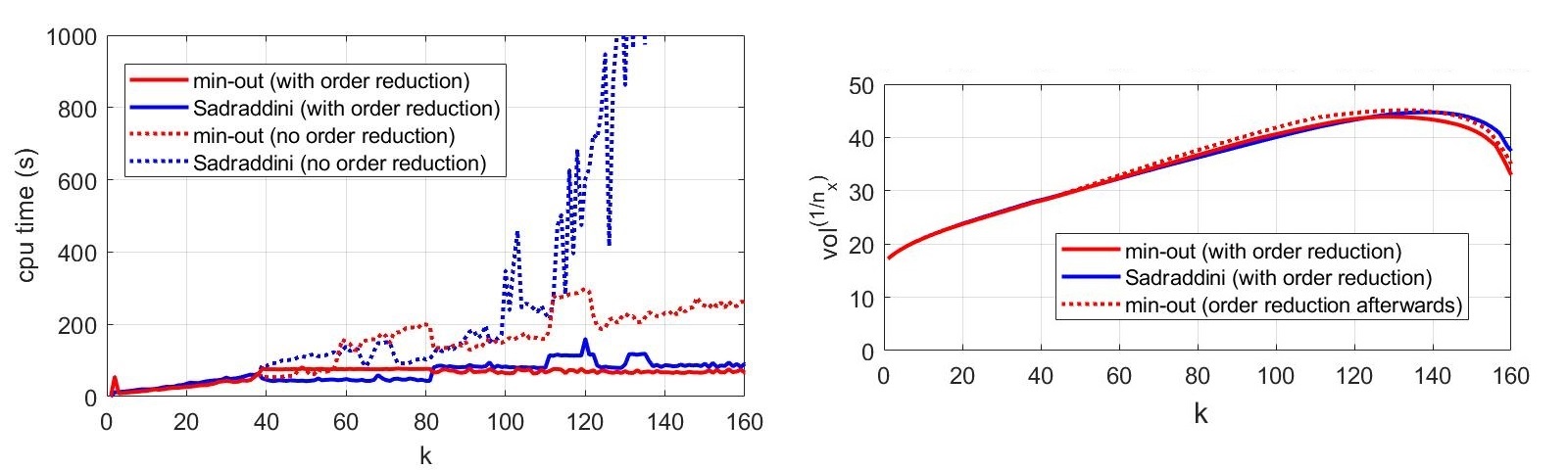}\\
  \caption{Backward reachable set computation for longitudinal dynamics. Left: computation time. Right: set volume.}\label{fig:long}
\end{figure*}

\subsection{Aircraft Position Control. }
With an aircraft position control system, we illustrate the overall BRS computation approach that combines Minkowski difference and order reduction to implement the iterations in Eq. \eqref{eq:cpreZapprx1}-\eqref{eq:cpreZapprx3}. 
The  linearized 6D lateral dynamics and the 6D longitudinal dynamics of the aircraft are in the form of Eq. \eqref{eq:sys}, whose $A$, $B$ matrices are given in Eq. \eqref{eq:ABlong}. 
For both systems, $E_{\rm lat} = E_{\rm long} = I$.
The states of the lateral and longitudinal dynamics are $x_{\rm lat} = [v, p, r, \phi, \psi, y]^\top$ and $x_{\rm long} = [u,w,q,\theta, x, h]^\top$ respectively, and control inputs are $u_{\rm lat} = [\delta_a, \delta_r]^\top$ and $u_{\rm long} = [\delta_e, \delta_t]^\top$ respectively (see TABLE \ref{tab:plane} and Fig. \ref{fig:ap}). 
We assume that the disturbance sets are hyper-boxes and their G-rep are $W_{\rm lat} = (\text{diag}([0.037, 0.00166, 0.0078, 0.00124, 0.00107,$  $0.07229]), 0)$ and $W_{\rm long} = (\text{diag}([0.3025,0.4025,0.01213, $ $0.006750,1.373,1.331]), 0)$.

\begin{align}
    A_{\rm lat} = &
%\left[
%\begin{array}{cccccc}
\begin{psmallmatrix} 
1.004     & 0.1408  & 0.3095 & -0.3112  & 0 & 0 \\
0.03015  & 1.177   & 0.6016 &  -0.6029 & 0 & 0 \\
-0.02448 & -0.1877& 0.3803 & 0.5642   & 0 & 0\\
-0.01057 & -0.09588 &  -0.3343 & 1.277 & 0 & 0\\
0.0003943 &  0.0095901 & -0.005341 & -0.007447 & 1 & 0\\
-0.2579 &  -23.32 & -51.03 & 61.35 & -37.86 & 1\\
\end{psmallmatrix}, 
%\end{array}
%\right]
%\end{array}
%\right]
% \label{eq:ABlat} 
\nonumber \\
    A_{\rm long} = & 
%\left[
%\begin{array}{cccccc}
\begin{psmallmatrix}
0.9911 & -0.04858 & -0.01709 & -0.4883 & 0 & 0 \\
0.0005870 & 0.9968 & 0.5168 & -0.0001398 & 0 & 0\\
0.0002070 & -0.001123 & 0.9936 & -5.092\times 10^{-5} & 0 & 0\\
1.907 & -1.032 & 0.01832 & 1 & 0 & 0 \\
-0.04601 & 0.001125 & 0.0002638 & 0.01130 & 1 & 0\\
-5.095\times 10^{-5} & -0.1874 & -0.01185 & 4.004 & 0 & 1\\
\end{psmallmatrix}, 
%\end{array}
%\right]
\nonumber \\
B_{\rm lat} =  & 
%\left[
%\begin{array}{cc}
\begin{psmallmatrix}
-0.1189 & 0.007812\\
-0.1217 & 0.2643\\
0.01773& -0.2219\\
-0.02882&-0.09982\\
-0.0005607&0.002437\\
0.1120&-0.5785\\
\end{psmallmatrix}, 
\nonumber \\
B_{\rm long} =  & 
%\left[
%\begin{array}{cc}
\begin{psmallmatrix}
1.504 & 7.349\times 10^{-5}\\
-0.04645 & -3.421 \times 10^{-6}\\
-0.009812&-1.488\times 10^{-6} \\
-9.080 \times 10^{-5}& -1.371\times 10^{-8}\\
-0.03479 & -1.700\times 10^{-6}\\
0.004171& 2.913\times 10^{-7}\\
\end{psmallmatrix}. 
%\end{array}
%\right]
\label{eq:ABlong} 
\end{align}

\begin{figure}[h]
  \centering
  \includegraphics[width=2.75in]{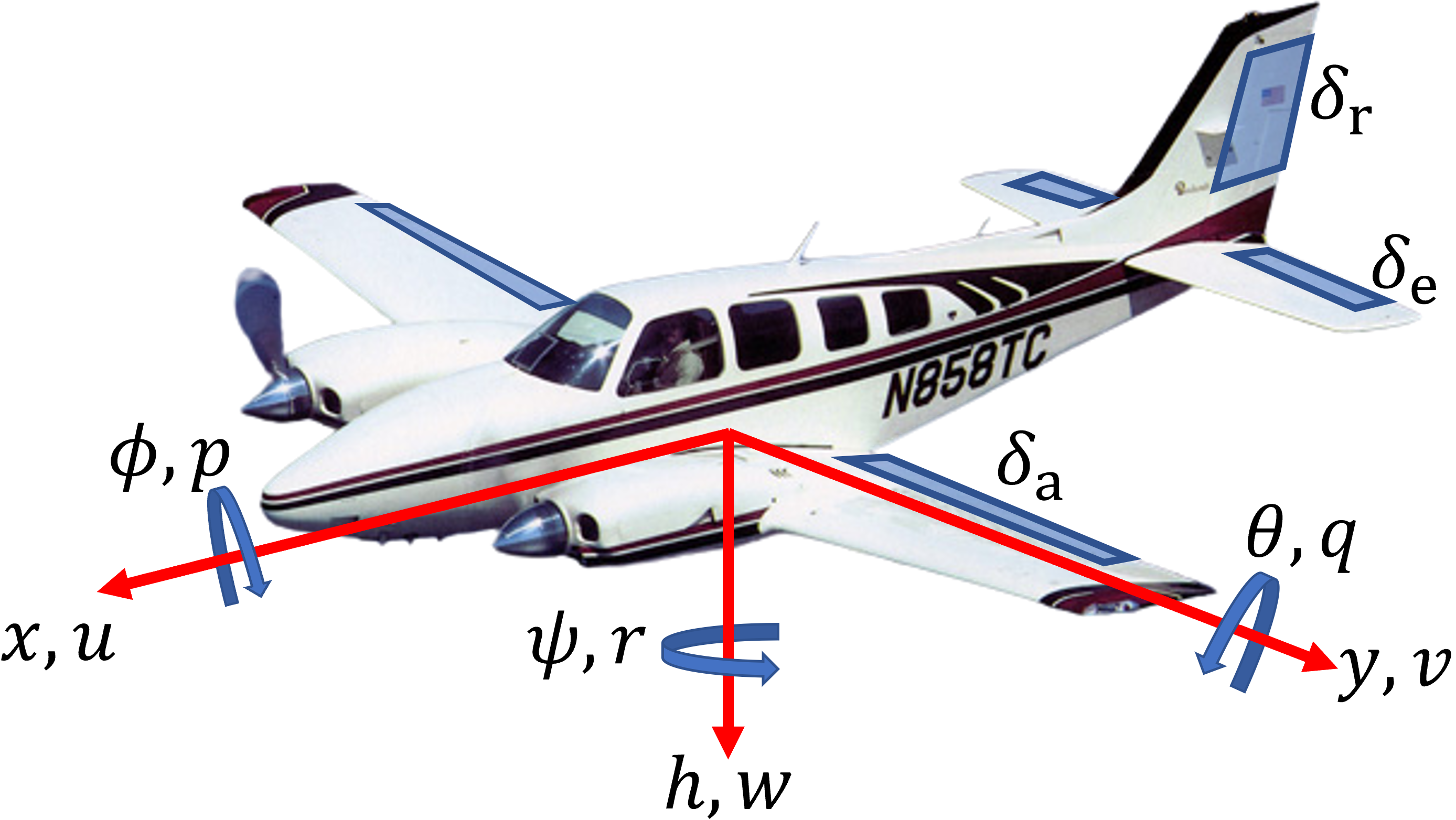}\\
  \caption{Illustration of the states and control inputs.}\label{fig:ap}
\end{figure}

\begin{table}[]
\centering
\caption{Variables in the aircraft model}
\begin{tabular}{c | c | c | c }
\hline
variable & physical meaning & range & unit \\ \hline
\hline
$v$ & velocity & $[-1, 1] $ & m/s\\  
$p$ & roll angular rate &$ [-1, 1] $& rad/s\\  
$r$ & yaw angular rate &$[-1, 1]$  & rad/s\\  
$\phi$ & roll angle & $[-\pi/5, \pi/5] $& rad\\  
$\psi$ & yaw angle & $[-\pi/5, \pi/5]$& rad\\  
$y$ & lateral deviation &$ [-2, 2]$ & m\\  
\hline
$u$ & velocity & $[40,60]$  & m/s\\ 
$w$ & velocity &$[0,10]$ & rad/s\\  
$q$ & pitch angular rate &$[-0.1,0.1]$   & rad/s\\  
$\theta$ & pitch angle & $[-\pi,\pi]$ & rad\\  
$x$ & horizontal displacement & $[0,800]$ & rad\\  
$h$ & altitude &$ [260 ,390]$ & m\\ 
\hline
$\delta_a$ & aileron deflection &$ [-\pi, \pi]$ & m\\ 
$\delta_r$ & rudder deflection &$ [-\pi, \pi]$ & m\\ 
\hline
$\delta_e$ & elevator deflection &$[-0.262,0.524]$ & m\\  
$\delta_t$ &throttle control &$[0,10^4]$ & m\\ \hline

\end{tabular}
\label{tab:plane}
\end{table}

For both the lateral and longitudinal dynamics, we can efficiently compute their $k$-step BRSs using the proposed approach for reasonably large horizons $k$, whereas the computation gets stuck at $k=3$ using the exact Minkowski difference provided by MPT3 \cite{MPT3}, or the approximation function implemented in CORA.  
Fig. \ref{fig:lat}, \ref{fig:long} show the results for the lateral dynamics and the longitudinal dynamics, respectively. 
In each figure, the left (right, resp.) plot shows  the cpu time for computing $\underline{Z}_k$ (the size of $\underline{Z}_k$, resp.) versus $k$, the number of backward expansion steps. 
The red curves are for our approach and the blue ones for the approach using Sadraddini’s zonotope containment encoding. 

%\begin{itemize}
%\item[1)]
\subsubsection{Cpu time plots} The solid (dotted, resp.) lines correspond to the computation time with (without, resp.) zonotope order reduction.
Using the order reduction technique (actived at $k = 39$), our approach and Sadraddini's approach give comparable results. 
Without order reduction, the computation time of Sadraddini's approach (dotted blue) grows faster w.r.t. $k$ than ours (dotted red). 
This is consistent with the big-O analysis because in our approach, the time-dominant Minkwoski difference step amounts to solving a linear program whose number of variables is linear in $k$ (proportional to the $\underline{Z}_k$'s order), whereas 
the number of variables is quadratic in $k$ in Sadraddini's formulation. 
Although our approach scales well even without order reduction, 
order reduction is still important in efficiently storing the zonotopic BRSs and 
deriving the control law. 

% \item[2)]
\subsubsection{Volume plots} The solid lines correspond to the results with order reduction ``in the loop" (i.e., in the $k^{\rm th}$ step, $\underline{Z}_k$ is reduced to a certain order before $\underline{Z}_{k+1}$ is computed). 
Whereas the dotted lines correspond to the results with order reduction after all $\underline{Z}_k$'s are computed (Ideally, we would like to compute $\underline{Z}_k$'s volume without order reduction at all, but this is impossible with the off-the-shelf volume computation tools in CORA because the complexity is  combinatorial  in $\underline{Z}_k$'s order).  
The two approaches give comparable results with order reduction. 
Moreover, since the dotted red line and the solid red line are close to each other, this indicates that the 
``wrapping effect" due to the order reduction in-the-loop is relatively small. 
% \end{itemize}

%For both the lateral and longitudinal dynamics, we can efficiently compute their $k$-step backward reachable sets using the proposed approach for reasonably large horizons $k$, whereas the computation gets stuck at $k=3$ using the exact Minkowski difference provided by MPT3 \cite{MPT3}, or the approximation function implemented in CORA.  
%The computation time and the volume of the backward reachable sets are plotted in Fig. \ref{fig:lat}, \ref{fig:long}.
%If the order reduction technique is used (active starting from $k = 39$), our approach and Sadraddini's approach give comparable results in computation time and sizes of the sets. Without order reduction, as the big-O analysis suggests, our approach scales better with $N$, the number of generators of $\underline{Z}_k$, which is proportional to $k$. 

\liren{
\subsection{Double Integrator with Uncontrollable Subspace}
With a 10D system, we show the effectiveness of the reachability controller derived from the zonotopic BRSs as described in Section \ref{sec:law}. 
The system consists of a double-integrator dynamics in the 3D space and a 4D uncontrollable subspace 
(the uncontrollable part 
% has complex eigenvalues and 
affects the controllable part). 
The continuous-time dynamics is 
\begin{align}
\dot{x}_1 & = x_2 + x_7 + x_{10} + w_1, \ \ \dot{x}_2 = u_1 + w_2, \nonumber \\
\dot{x}_3 & = x_4 - x_8 + w_3,  \ \ \dot{x}_4 = u_2 + w_4,  \nonumber \\
\dot{x}_5 & = x_6 + x_9 + w_5,  \ \ \dot{x}_6 = u_3 + w_6,   \\
\dot{x}_7 & = -0.01x_7 + x_8 + w_7,  \ \ \dot{x}_8 = - x_8 - 0.01x_7 + w_8,  \nonumber \\
\dot{x}_9 & = -10^{-4}x_7 +2 x_{10} + w_9, \nonumber \\
\dot{x}_{10} & = - 2 x_9 - 10^{-4}x_{10} + w_{10}.  \nonumber  
\end{align}
We discretize the above dynamics with a sampling period $\Delta t = 0.5$s, and define the disturbance set $W$ so that 
$w_{\{1,3,5\}} \in [-0.12, 0.12]$, $w_{\{2,4,6\}} \in [-0.2, 0.2]$, $w_{\{7,8,9,10\}} \in [-0.1, 0.1]$,
%\begin{align}
%w_i \in 
%\begin{cases}
%[-0.12, 0.12],  &  i \in\{1,3,5\}  \\
%[-0.2, 0.2],  & i \in \{2,4,6\} \\
%[-0.1, 0.1],  & i \in \{7,8,9,10\}  
%\end{cases}, 
%\end{align}
and the control set $U = [-0.5, 0.5]^3$. %$U = [-0.5, 0.5]\times [-0.5, 0.5]\times [-0.5, 0.5]$. 
Starting from a randomly picked initial condition in $\underline{Z}_{50}$, our goal is to reach a final state for which $x_i \in [9.5, 10.5]$ for $i \in \{1,3,5\}$ and $x_i \in [-0.5, 0.5]$ for the remaining $i$'s.  
We defined a controller as described in Section \ref{sec:law}. 
Fig. \ref{fig:int} shows a closed-loop trajectory under random disturbances.
The small target set is reached despite the oscillating uncontrollable dynamics.

\begin{figure}[h]
  \centering
  \includegraphics[width=2.85in]{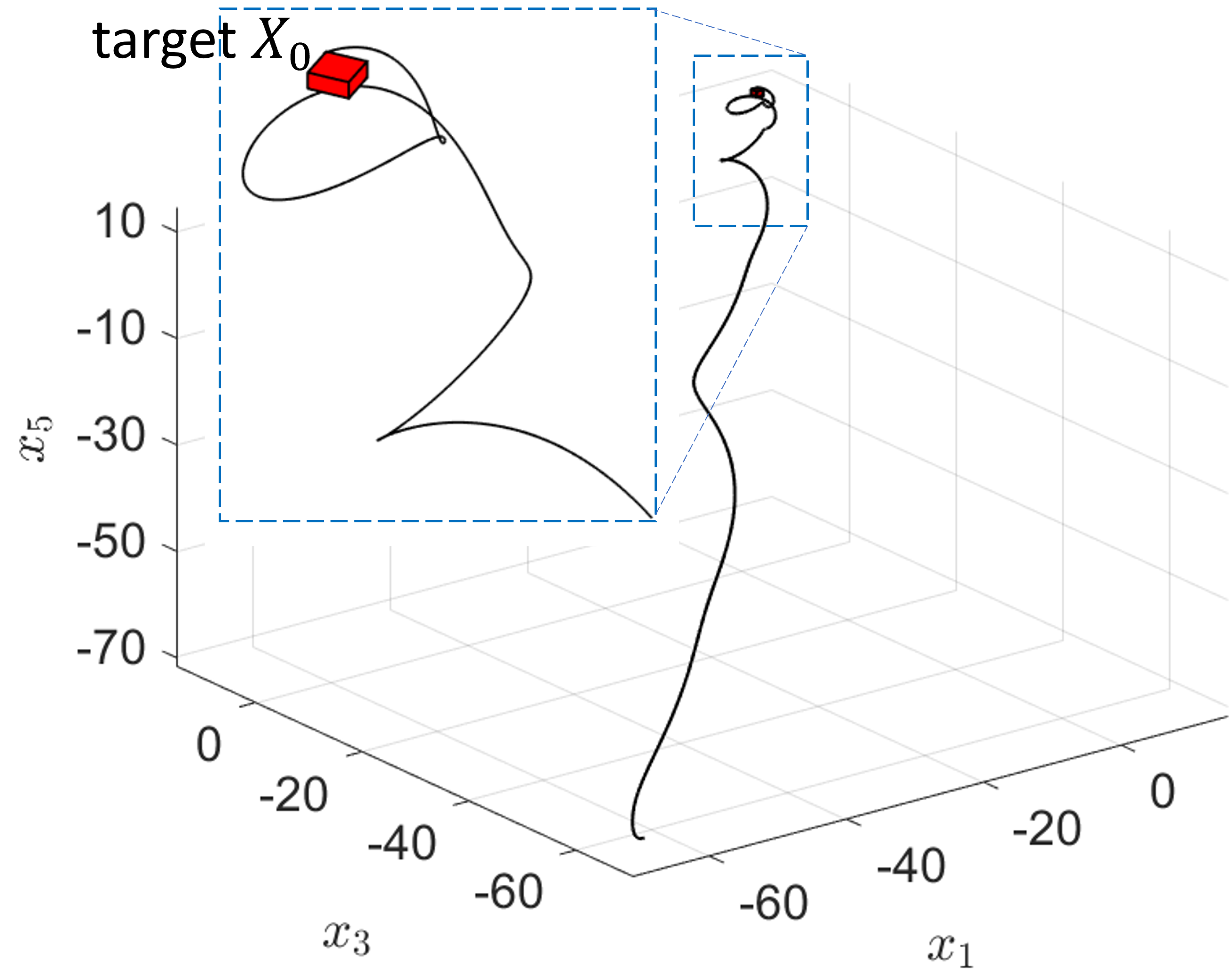}\\
  \caption{A closed-loop trajectory for the double-integrator dynamics. The red box is the target set. }\label{fig:int}
\end{figure}

\section{Conclusion}
% In this paper, we investigate using zonotopes to represent and under-approximate BRSs for uncertain linear systems. 
In this paper, we develop an approach that under-approximates the backward reachable sets for uncertain linear systems using zonotopes.} 
The main technical ingredients are i) under-approximating the Minkowski difference between two zonotopes and ii) an order reduction technique tailored to enclosed zonotopes. 
These developments were evaluated with randomly generated instances and two case studies. 
Experiments show that our method is more scalable than the off-the-shelf tools (MPT3, CORA) and scales differently from the approach based on Sadraddini's zonotope-inclusion technique. 
In our method, the dominant Minkowski subtraction step requires solving a linear program whose size is linear in the zonotope's order, while that dependency is quadratic in Sadraddini's approach.
We will investigate extending our approach to nonlinear systems in the future. 

% {\color{red}to end in a positive note, explicitly  say ""  You were also ssaying something about $W$ vertex size not changing with the backward reachable set horizon whereas some of the things sadraddini's method depends on changes with horizon length. is that discussion in the paper? I might have missed it.}

\noindent{\em Acknowledgments:} The authors would like to thank Yuhao Zhang from the University of Wisconsin-Madison and Sara Shoouri and Jiahong Xu from the University of Michigan for sharing the aircraft model.

\balance

\bibliographystyle{abbrv}
\bibliography{IEEEabrv,main}

\begin{thebibliography}{10}

\bibitem{althoff2015introduction}
M.~Althoff.
\newblock An introduction to {CORA} 2015.
\newblock In {\em 1st and 2nd Intl. Workshop on ARCH}, pages 120--151, 2015.

\bibitem{althoff2015computing}
M.~Althoff.
\newblock On computing the minkowski difference of zonotopes.
\newblock {\em arXiv preprint arXiv:1512.02794}, 2015.

\bibitem{althoff2021set}
M.~Althoff, G.~Frehse, and A.~Girard.
\newblock Set propagation techniques for reachability analysis.
\newblock {\em Annual Review of Control, Robotics, and Autonomous Systems},
  4:369--395, 2021.

\bibitem{bertsekas1972infinite}
D.~Bertsekas.
\newblock Infinite time reachability of state-space regions by using feedback
  control.
\newblock {\em IEEE TAC}, 17(5):604--613, 1972.

\bibitem{blanchini2008set}
F.~Blanchini and S.~Miani.
\newblock {\em Set-theoretic methods in control}.
\newblock Springer, 2008.

\bibitem{chen2018signal}
M.~Chen, Q.~Tam, S.~C. Livingston, and M.~Pavone.
\newblock Signal temporal logic meets reachability: Connections and
  applications.
\newblock In {\em Intl. WAFR}, pages 581--601. Springer, 2018.

\bibitem{chou2018using}
G.~Chou, Y.~E. Sahin, L.~Yang, K.~J. Rutledge, P.~Nilsson, and N.~Ozay.
\newblock Using control synthesis to generate corner cases: A case study on
  autonomous driving.
\newblock {\em IEEE TCAD}, 37(11):2906--2917, 2018.

\bibitem{girard2005reachability}
A.~Girard.
\newblock Reachability of uncertain linear systems using zonotopes.
\newblock In {\em Intl. Workshop on HSCC}, pages 291--305. Springer, 2005.

\bibitem{gover2010determinants}
E.~Gover and N.~Krikorian.
\newblock Determinants and the volumes of parallelotopes and zonotopes.
\newblock {\em Linear Algebra Its Appl.}, 433(1):28--40, 2010.

\bibitem{han2016enlarging}
D.~Han, A.~Rizaldi, A.~El-Guindy, and M.~Althoff.
\newblock On enlarging backward reachable sets via zonotopic set membership.
\newblock In {\em ISIC}, pages 1--8. IEEE, 2016.

\bibitem{MPT3}
M.~Herceg, M.~Kvasnica, C.~Jones, and M.~Morari.
\newblock {Multi-Parametric Toolbox 3.0}.
\newblock In {\em Proc.~of ECC}, pages 502--510, Z\"urich, Switzerland, July
  17--19 2013.
\newblock \url{http://control.ee.ethz.ch/~mpt}.

\bibitem{kochdumper2020computing}
N.~Kochdumper and M.~Althoff.
\newblock Computing non-convex inner-approximations of reachable sets for
  nonlinear continuous systems.
\newblock In {\em 59th CDC}, pages 2130--2137. IEEE, 2020.

\bibitem{kopetzki2017methods}
A.-K. Kopetzki, B.~Sch{\"u}rmann, and M.~Althoff.
\newblock Methods for order reduction of zonotopes.
\newblock In {\em 56th CDC}, pages 5626--5633. IEEE, 2017.

\bibitem{kurzhanskiy2011reach}
A.~A. Kurzhanskiy and P.~Varaiya.
\newblock Reach set computation and control synthesis for discrete-time
  dynamical systems with disturbances.
\newblock {\em Automatica}, 47(7):1414--1426, 2011.

\bibitem{lasserre2015tractable}
J.~B. Lasserre.
\newblock Tractable approximations of sets defined with quantifiers.
\newblock {\em Mathematical Programming}, 151(2):507--527, 2015.

\bibitem{li2019robustly}
Y.~Li.
\newblock {\em Robustly complete temporal logic control synthesis for nonlinear
  systems}.
\newblock PhD thesis, University of Waterloo, 2019.

\bibitem{li2017invariance}
Y.~Li and J.~Liu.
\newblock Invariance control synthesis for switched nonlinear systems: An
  interval analysis approach.
\newblock {\em IEEE TAC}, 63(7):2206--2211, 2017.

\bibitem{liebenwein2020compositional}
L.~Liebenwein, W.~Schwarting, C.-I. Vasile, J.~DeCastro, J.~Alonso-Mora,
  S.~Karaman, and D.~Rus.
\newblock Compositional and contract-based verification for autonomous driving
  on road networks.
\newblock In {\em Robotics Research}, pages 163--181. Springer, 2020.

\bibitem{lygeros1999controllers}
J.~Lygeros, C.~Tomlin, and S.~Sastry.
\newblock Controllers for reachability specifications for hybrid systems.
\newblock {\em Automatica}, 35(3):349--370, 1999.

\bibitem{mitchell2007comparing}
I.~M. Mitchell.
\newblock Comparing forward and backward reachability as tools for safety
  analysis.
\newblock In {\em Intl. Workshop on HSCC}, pages 428--443. Springer, 2007.

\bibitem{mitchell2005time}
I.~M. Mitchell, A.~M. Bayen, and C.~J. Tomlin.
\newblock A time-dependent hamilton-jacobi formulation of reachable sets for
  continuous dynamic games.
\newblock {\em IEEE TAC}, 50(7):947--957, 2005.

\bibitem{montejano1996some}
L.~Montejano.
\newblock Some results about minkowski addition and difference.
\newblock {\em Mathematika}, 43:265--273, 1996.

\bibitem{sadraddini2019linear}
S.~Sadraddini and R.~Tedrake.
\newblock Linear encodings for polytope containment problems.
\newblock In {\em 58th CDC}, pages 4367--4372. IEEE, 2019.

\bibitem{smith2016interdependence}
S.~W. Smith, P.~Nilsson, and N.~Ozay.
\newblock Interdependence quantification for compositional control synthesis
  with an application in vehicle safety systems.
\newblock In {\em 55th CDC}, pages 5700--5707. IEEE, 2016.

\bibitem{yang2020efficient}
L.~Yang and N.~Ozay.
\newblock Efficient safety control synthesis with imperfect state information.
\newblock In {\em 59th CDC}, pages 874--880. IEEE, 2020.

\bibitem{yang2018comparison}
X.~Yang and J.~K. Scott.
\newblock A comparison of zonotope order reduction techniques.
\newblock {\em Automatica}, 95:378--384, 2018.

\end{thebibliography}

\end{document}